\documentclass{article}

\usepackage{arxiv}

\newif\ifappendix
\appendixtrue

\usepackage[utf8]{inputenc} 
\usepackage[T1]{fontenc}    
\usepackage{color}       
\usepackage{hyperref}       
\usepackage{url}            
\usepackage{booktabs}       
\usepackage{amsmath}       %
\usepackage{amsthm}       %
\usepackage{amsfonts, dsfont}       
\usepackage{nicefrac}       
\usepackage{microtype}      
\usepackage{algorithm,algpseudocode}
\usepackage{xspace}
\usepackage{enumitem}  
\usepackage{balance}
\usepackage{makecell}
\usepackage{subfigure} 
\usepackage{graphicx, bbm}  

\usepackage{tikz}
\usetikzlibrary{shapes,snakes}
\usetikzlibrary{arrows,automata}
\usetikzlibrary{positioning}

\usepackage[utf8]{inputenc}
\usepackage{enumitem}

\newtheorem{thm}{Theorem}[section]

\newtheorem{lemma}[thm]{Lemma}

\newcommand{\cA}{\ensuremath{\mathcal{A}}}

\newcommand{\ZZP}{\mathbb{Z}_{+}}

\newcommand{\OPT}{\mbox{\sc OPT}}

\newcommand{\expct}[1]{\ensuremath{\text{{\bf E}$\left[#1\right]$}}}
\newcommand{\size}[1]{\ensuremath{\left|#1\right|}}

\mathchardef\mhyphen="2D
\newcommand{\cost}{\textsc{cost}}
\newcommand{\fcost}{\textsc{fcost}}
\newcommand{\mcost}{\textsc{mcost}}

\newcommand{\inn}{\fcost^{in}}
\newcommand{\out}{\fcost^{out}}

\newcommand{\CC}{\mathcal{C}}
\newcommand{\CP}{\mathcal{P}}
\newcommand{\CF}{\mathcal{F}}
\newcommand{\Intra}{\mathrm{intra}}
\newcommand{\Inter}{\mathrm{inter}}

\newcommand{\pivot}{{\sc Pivot}\xspace}
\newcommand{\match}{{\sc Match}\xspace}
\newcommand{\repmatch}{{\sc Rep. Match}\xspace}
\newcommand{\localsearch}{{\sc Local}\xspace}
\newcommand{\randbaseline}{{\sc Rand}\xspace}
\newcommand{\single}{{\sc Single}\xspace}
\newcommand{\error}{{\sc Error}\xspace}
\newcommand{\imbalancecolor}{{\sc Imbalance}\xspace}

\newcommand{\para}[1]{\smallskip \noindent{\bf #1.}}

\newenvironment{lemmafirst}[1]{\begin{lemma}\label{#1}
}{\end{lemma}}
\newenvironment{lemmaagain}[1]{\begin{trivlist}\label{#1_appendix}
\item[] \textbf{Lemma~\ref{#1}.} \em}{\end{trivlist}}

\newenvironment{thmfirst}[1]{\begin{thm}\label{#1}
}{\end{thm}}
\newenvironment{thmagain}[1]{\begin{trivlist}\label{#1_appendix}
\item[] \textbf{Theorem~\ref{#1}.} \em}{\end{trivlist}}

\title{Fair Correlation Clustering}

\author{
Sara Ahmadian \\ Google Research \\ \texttt{sahmadian@google.com} 
\And Alessandro Epasto \\ Google Research \\ \texttt{aepasto@google.com} 
\And  Ravi Kumar \\ Google Research \\ \texttt{ravi.k53@gmail.com}  
\And Mohammad Mahdian \\ Google Research \\ \texttt{mahdian@google.com} }

\begin{document}

\maketitle

\begin{abstract}\footnote{To appear in Proceedings of the 23rd International Conference on Artificial Intelligence and Statistics (AISTATS) 2020, Palermo, Italy. PMLR: Volume 108. Copyright 2020 by the author(s)} In this paper, we study correlation clustering under fairness constraints.  Fair variants of $k$-median and $k$-center clustering have been studied recently, and approximation algorithms using a notion called fairlet decomposition have been proposed.  We obtain approximation algorithms for fair correlation clustering under several important types of fairness constraints.

Our results hinge on obtaining a fairlet decomposition for correlation clustering by introducing a novel combinatorial optimization problem.  We define a fairlet decomposition with cost similar to the $k$-median cost and this allows us to obtain approximation algorithms for a wide range of fairness constraints. 

We complement our theoretical results with an in-depth analysis of our algorithms on real graphs where we show that fair solutions to correlation clustering can be obtained with limited increase in cost compared to the state-of-the-art (unfair) algorithms. 

\end{abstract}

\section{Introduction}
There is a growing literature on fairness in various learning and optimization problems~\cite{kamishima2011fairness,kamishima2012fairness,joseph2016fairness,celis2017ranking,yang2017measuring,celis2018multiwinner,matroidfair}. The goal of this literature is to develop criteria and algorithms to ensure that we can find solutions for  optimization/learning problems that are fair with respect to a certain sensitive feature. In the case of clustering, a fundamental unsupervised learning and optimization problem, the study of fairness was initiated by Chierichetti et al.~\cite{chierichetti2017fair}. They formulated the notion of proportional fairness, and developed approximation algorithms for fair $k$-median and fair $k$-center under this notion of fairness. Follow-up work generalized their results to other clustering problems, such as $k$-means and facility location, and to  more relaxed notions of fairness~\cite{bercea2018fair,kdd19,bera2019fair,scalablefair,kcenterfair}. Notably, the important graph problem of correlation clustering has been so far not addressed by this literature.

Correlation clustering uses information about both similarity and dissimilarity relationships among a set of objects in order to cluster them~\cite{Bansal2004}. In contrast to other clustering problems such as $k$-median, $k$-means, and $k$-center, the number of clusters is not pre-specified but rather determined based on the outcome of an optimization. This, as well as the fact that correlation clustering uses both similarity and dissimilarity information, makes it a desirable clustering model in many applications~\cite{kushagra2018semi,pouget2019variance,acn}. Therefore, it is natural to study this problem under fairness constraints.

The main tool introduced by Chierichetti et al.~\cite{chierichetti2017fair} for solving fair $k$-center and $k$-median problems is the notion of {\em fairlets}. A fairlet is a small set of elements that satisfies the fairness property. Chierichetti et al.~\cite{chierichetti2017fair} showed that fair $k$-median and $k$-center can be solved by first decomposing an instance into fairlets and then solving the clustering problem on the set of centers of these fairlets. To the best of our knowledge, this technique has been used only for {\it metric space} clustering problems such as $k$-center and $k$-median.

Our main result is developing a fairlet-based reduction for the {\it graph} clustering problem of correlation clustering. Whereas, in the case of $k$-center and $k$-median, the fairlet decomposition problem amounts to solving the same clustering problem on the same instance under the condition that each cluster is a fairlet, the situation for correlation clustering is complicated by the lack of the properties of metric spaces. To tackle this problem, we introduce a novel cost function for the correlation clustering fairlet decomposition, and prove that this cost can be approximated by a median-type clustering cost function for a carefully defined metric space. 

Given a solution to this fairlet decomposition problem, we show that the fair correlation clustering instance can be reduced to a regular correlation clustering instance through a graph transformation.  Therefore, any approximation algorithm for fairlet decomposition with median cost yields an approximation algorithm for fair correlation clustering; the loss in the approximation ratio depends on the size of the fairlets.  We show that in many natural cases, there is a fairlet decomposition with small fairlets, thereby bounding the approximation ratio of our algorithm for fair correlation clustering. 

In addition to the theoretical bounds on our algorithm, we provide an empirical evaluation based on real data sets, showing that the algorithms often perform much better in practice than their worst-case guarantees and that they yield solutions of costs comparable to that of unfair clustering algorithms while substantially reducing the cluster imbalance.

\para{Related work}
Clustering is a fundamental unsupervised machine learning task with a long history (cf.~\cite{jain2010data}). Our paper spans the areas of correlation clustering, clustering in metric spaces, and fairness in clustering which are actively growing fields. For brevity, we will only focus of key works in these three areas.

\noindent {\it Correlation clustering.}
Correlation clustering is a widely studied formulation of clustering with both similarity and dissimilarity information~\cite{Bansal2004}, with many applications in machine learning~\cite{kushagra2018semi,bressan2019correlation}. Variants of the problem include complete signed graphs~\cite{Bansal2004,acn} and weighted graphs~\cite{demaine2006correlation}. We focus on the complete graph case with $\pm 1$ weights which is APX-hard~\cite{charikar2005clustering} but admits constant-factor algorithms~\cite{acn,CMSY}. Distributed and streaming algorithms are also known~\cite{pan2015parallel,ahn2015correlation}.

\noindent {\it Metric space clustering.} 
The most widely studied clustering setting is clustering in metric spaces consisting in minimizing the $\ell_p$-norm of the distances between points in a cluster and their center. For $p\in\{1,2,\infty\}$ this corresponds to $k$-median, $k$-means, and $k$-center, respectively, which are NP-hard problems but admit constant-factor approximations~\cite{gonzalez1985clustering, hochbaum1985best,DBLP:journals/siamcomp/LiS16, DBLP:conf/focs/AhmadianNSW17,kanungo2004local}.

\noindent {\it Fairness in clustering.}
Fairness in machine learning is new area with a fast growing literature. Fundamental work in this area is devoted to defining notions of fairness~\cite{calders2010three, dwork2012fairness, feldman2015certifying, kamishima2012fairness} and solving fairness-constrained problems~\cite{celis2018multiwinner, celis2017ranking,chierichetti2017fair, joseph2016fairness,kamishima2012fairness, yang2017measuring,scalablefair,kamishima2011fairness,fish2016confidence,kdd19,matroidfair,har2019near}.

Chierichetti et al.~\cite{chierichetti2017fair} first introduced a notion of disparate impact for clustering and provided fair $k$-center algorithms for the case of two colors (or groups); see Section~\ref{sec:model}. Following this work, the problem has been later generalized in many directions including allowing many colors~\cite{rosner2018privacy}, allowing upper bounds on the fraction of points of a given color~\cite{kdd19} and both upper and lower bounds~\cite{bera2019fair,bercea2018fair}. Backurs et al.~\cite{
scalablefair} designed near-linear algorithms for finding $k$-median fairlets, and Huang et al.~\cite{huang2019coresets}  designed  core-sets for the problem. Other variants include clustering with diversity constraints~\cite{liyizhang}, proportionality constraints~\cite{chen2019proportionally}, and fair center selection~\cite{kcenterfair}. 
Fairness has been studied in spectral clustering as well~\cite{spectralfair,ziko2019clustering}. From an application point, fair clustering can be seen through the lenses of fair allocation~\cite{elzayn2019fair}, Medicaid eligibility~\cite{fang2019achieving}, and ensuring protected group representations~\cite{dash2019summarizing}.

To the best of our knowledge no prior work has addressed correlation clustering with fairness constraints in the cluster elements distribution. Kalhan~\cite{kalhan2019improved} recently studied a fairness notion in correlation clustering in which the maximum error for a vertex is bounded.

\section{Problem Statement}
\label{sec:model}

\para{Correlation clustering}
Let $G = (V, E)$ be a complete undirected graph on $|V| = n$ vertices and $\sigma: E\mapsto \mathbb{R}$ be a function that assigns a label to each edge. The label $\sigma(e)$ for each $e$ is either positive (indicating that the two endpoints of $e$ are \emph{similar}) or non-positive (indicating that they are \emph{dissimilar}). In the {\em unweighted} version of the problem~\cite{Bansal2004}, $\sigma(e)\in\{-1,+1\}$ for each $e$.
Our focus in this paper is on the unweighted version, although we will use the weighted version in the proofs. Let $E^+ = \{e\in E~\mid~\sigma(e)>0\}$ be the set of positive edges and $E^- = E \setminus E^+$ be the set of non-positive edges.  For subsets $S, T \subseteq V$, let $E(S) = E \cap S^2$ denote the edges inside $S$ and $E(S, T) = E \cap (S \times T)$ denote the edges between $S$ and $T$.  Let $E^+(S, T) = E^+ \cap E(S, T)$ and
$E^-(S, T) = E^- \cap E(S, T)$.

A \emph{clustering} is a partitioning $\CC = \{C_1, C_2, \ldots\}$ of $V$ into disjoint subsets. The sets of intra-cluster and inter-cluster edges in a clustering $\CC$ are defined as $\Intra(\CC) = \bigcup_{C\in\CC}E(C)$ and $\Inter(\CC) = E \setminus \Intra(\CC)$. 
The \emph{correlation clustering cost} of $\CC$ is defined as:
\[
\cost(G, \CC) = \sum_{e\in \Intra(\CC)\cap E^-} |\sigma(e)| + \sum_{e\in \Inter(\CC)\cap E^+} |\sigma(e)|.
\]
In the unweighted version of the problem, this is simply 
$\cost(G,\CC) = |\Intra(\CC)\cap E^-| + |\Inter(\CC)\cap E^+|$.
The goal of correlation clustering\footnote{This is the minimizing disagreements variant of correlation clustering. A maximizing agreements version can also be defined similarly. In this paper we focus on minimizing disagreements, since the maximization version admits a trivial randomized $2$-approximation that can be made fair.} is to find a clustering $\CC$ to minimize $\cost(G, \CC)$. For unweighted correlation clustering, there are constant-factor approximation algorithms for this problem~\cite{Bansal2004,acn,CMSY}, with the best known constant being 2.06.  For the weighted version, the best known algorithm obtains an $O(\log n)$-approximation~\cite{demaine2006correlation}. 

\para{Fairness constraints}
In the fair version of any clustering problem, each vertex $v \in V$ has a color $c(v)$.  Proportional fairness, defined by Chierichetti et al.~\cite{chierichetti2017fair}, requires that in every cluster, the number of vertices of each color is proportional to the corresponding number in the whole graph. In particular, in the symmetric case where each color appears the same number of times in the graph, we require the same in each cluster. Ahmadian et al.~\cite{kdd19} relaxed this property by requiring that each color constitutes at most an $\alpha$-fraction of each cluster, for a given $\alpha \in (0, 1)$. Bera et al.~\cite{bera2019fair} further generalized this notion to  include lower bounds on the number of vertices of each color in each cluster.

We give a general reduction from fair correlation clustering to a median fairlet decomposition that works for any of these definitions of fairness, and in fact for a more general class of constraints. As long as the fairlet decomposition problem can be solved with small fairlets (holds for the above fairness definitions; see Section~\ref{sec:decomposition}), this will give us an approximation algorithm for the corresponding fair correlation clustering problems.

\section{Overview of Results} \label{sec:overview}
In this section, we give a high-level overview of our algorithm and our main result.  A key ingredient of our algorithm is a general reduction from the given constrained correlation clustering problem (as defined below) to a {\em fairlet decomposition} problem. We then show how the cost of a fairlet decomposition can be approximated by a median clustering cost function. This allows us to use previous results on the fair median problem to solve fairlet decomposition for the standard notions of fairness defined in the previous section. Finally, given an approximately optimal fairlet decomposition, we use our reduction to reduce the constrained correlation clustering instance to a standard correlation clustering instance, and apply known algorithms~\cite{Bansal2004,acn,CMSY} to solve this problem.

\para{Constrained correlation clustering}
We start by defining a general class of \emph{constrained correlation clustering} problems. Consider an unweighted correlation clustering instance $G$ and let $\CF$ be a family of subsets of $V$. We treat $\CF$ as the family of feasible clusters, and assume it has the following {\em composability} property: for every $F_1, F_2\in\CF$, we have $F_1\cup F_2\in\CF$. Note this property is satisfied when $\CF$ is the collection of all fair sets under any of the definitions of fairness given in Section~\ref{sec:model}. The constrained correlation clustering problem is to define a correlation clustering $\CC$ with minimum $\cost(G,\CC)$ such that for all $C\in\CC$, we have $C\in\CF$.

\paragraph{Fairlet decomposition.} Next, we define the notion of {\em fairlet decomposition} used in our reduction. A fairlet decomposition for a constrained correlation clustering problem is simply a partition $\CP = \{P_1, P_2, \ldots\}$ of $V$ into subsets in $\CF$, i.e., $P_i\in\CF$ for all $i$. We call each $P_i$ a {\em fairlet}. The key in our reduction is a cost function $\fcost$ that evaluates $\CP$'s usefulness in building a correlation clustering of $G$. Here we define this cost function, and in Section~\ref{sec:median} we show how it can be approximated by the standard median clustering cost function in a carefully defined metric space.

\paragraph{Fairlet decomposition cost.}
Consider a fairlet decomposition $\CP = \{P_1, P_2, \ldots\}$. For each fairlet $P_i$, we let $\inn(P_i)$ be the number of negative edges inside $P_i$, i.e., 
$\inn(P_i)= \size{E^- \cap \Intra(P_i)}$. For fairlets $P_i, P_j$, we let
let $\out(P_i, P_j)$ be the number of edges between them with the minority sign, i.e., 
\[
\out(P_i, P_j) = \min(\size{E^-(P_i, P_j)}, \size{E^+(P_i, P_j)}).
\]
Finally, we let $\inn(\CP) = \sum_{i}\inn(P_i)$,  $\out(\CP)  = \sum_{i<j} \out(P_i, P_j)$, and $\fcost(\CP) = \inn(\CP) + \out(\CP)$.

\paragraph{Reduced instance.} Given a constrained correlation clustering instance $G$ and a fairlet decomposition $\CP$ for $G$, we define a \emph{reduced} correlation clustering instance as follows.  Let $G^\CP$ be a complete graph on $\{p_1, \ldots, p_{|\CP|} \}$, where  vertex $p_i$ corresponds to fairlet $P_i \in \CP$.  The label $\sigma(p_i, p_j)$ of the edge between $p_i$ and $p_j$ is the majority sign of the edges in $E(P_i, P_j)$ (with ties broken arbitrarily) multiplied by a weight that is equal to the number of edges in $E(P_i, P_j)$ with the majority sign. 
 
Note that the instance $G^\CP$ defined above is an instance of {\em weighted} correlation clustering, although as we will observe, the edges have weights that are within a constant factor each other, and therefore the problem can be solved using unweighted correlation clustering algorithms. Given a solution to this problem, it can be expanded into a solution of the original constrained problem. The final algorithm is sketched below.

\begin{algorithm}[h]
\begin{algorithmic}[1]
\State $\CP \gets $ approx. fairlet decomp. (Lemmas~\ref{lem:fair-decom-half},~\ref{lem:fair-decom-eqaul}).
\State $G^\CP$: $(p_i, p_j)$ gets majority sign in $E(P_i, P_j)$ and weight $\max(|E^+(P_i, P_j), |E^-(P_i, P_j)|)$.  
\State Let $\CC$ be an approximate (non-constrained) correlation clustering solution of $G^\CP$.
\State Output the clustering $\{\bigcup_{p_j\in C_i} P_j: C_i\in\CC\}$.
\end{algorithmic}
\caption{Constrained Correlation Clustering}
\label{alg:main}
\end{algorithm}

To prove that the Algorithm~\ref{alg:main} produces an approximately optimal solution to the constrained correlation clustering problem, we need the following lemmas. The first two lemmas prove that a solution of $G$ can be transformed to a solution of $G^\CP$ and vice versa, and these transformations do not increase the cost by more than the cost of the fairlet decomposition. The third lemma bounds the cost of a fairlet decomposition in terms of the cost of the optimal solution to the constrained correlation clustering problem. 
\ifappendix The proofs of these lemmas are presented in Appendix~\ref{sec:proofs1}.

\begin{lemmafirst}{lem:1}
Given a correlation clustering instance $G$, a fairlet decomposition $\CP$ for $G$, and a clustering $\CC$ of $G$, there exists a clustering $\CC'$ of $G^\CP$ such that
\[
    \cost(G^\CP,\CC') \leq \cost(G,\CC) + \out(\CP).
\]
\end{lemmafirst}

\begin{lemmafirst}{lem:2}
Let $\CC$ be a clustering of $G^{\CP}$ and $\CC'$ be the clustering computed in line 4 of Algorithm~\ref{alg:main}.  Then,
\[
\cost(G,\CC') \leq \cost(G^\CP,\CC) + \fcost(\CP).
\]
\end{lemmafirst}

\begin{lemmafirst}{lem:3}
For any constrained correlation clustering instance $G$, and any constrained clustering $\CC$ of $G$, there is a fairlet decomposition $\CP$ of $G$ satisfying $\fcost(\CP)\le \cost(G,\CC)$.
\end{lemmafirst}

Putting these together, we have the following:

\begin{thm}
\label{thm:fairletreduction}
Assume there is an $\eta$-approximation algorithm $A_1$ for finding the minimum cost fairlet decomposition $\CP$ and a $\beta$-approximation algorithm $A_2$ for solving the unconstrained correlation clustering instance $G^\CP$. Then Algorithm~\ref{alg:main} produces a $(\beta(1+\eta)+\eta)$-approximation for the constrained correlation clustering instance $G$.
\end{thm}
\begin{proof}
Let $\OPT$ be an optimal solution to the constrained correlation clustering instance $G$. By Lemma~\ref{lem:3}, the fairlet decomposition problem has a solution of cost at most $\cost(G, \OPT)$, and therefore, algorithm $A_1$ for this problem must find a decomposition $\CP$ with $\fcost(\CP)\le \eta\cdot\cost(G, \OPT)$. Also, by Lemma~\ref{lem:1}, the instance $G^\CP$ has a solution of cost at most $(1+\eta) \cdot \cost(G, \OPT)$. Therefore, algorithm $A_2$ can find a clustering $\CC$ of cost at most $\beta(1+\eta) \cdot \cost(G, \OPT)$. Thus, by Lemma~\ref{lem:2}, the cost of the clustering produced by Algorithm~\ref{alg:main} is at most $(\beta(1+\eta)+\eta) \cdot \cost(G, \OPT)$. Finally, by the composability property of the constraints, we know that this clustering satisfies the constraints, since each of its clusters is a union of fairlets in $\CP$.
\end{proof}

In the following, we explain the approximation factor $\beta$ we can get for solving unconstrained correlation clustering instance $G^\CP$ and we dedicate the next section to approximation ratios $\eta$ that we can get for minimum cost fairlet decomposition problem depending on the fairness parameter $\alpha$ and the number of colors in a given fair correlation clustering instance. 

\begin{lemma}\label{lem:unconstraited-lem}
There exists an approximation algorithm for 
unconstrained correlation clustering of $G^\CP$  with approximation ratio of $\beta = \min(\log n, 2\rho r^2)$ where $r = \frac{\max_{P\in \CP} |P|}{\min_{P\in \CP} |P|}$ and $\rho$ is the approximation factor of unweighted correlation clustering\footnote{Currently best known approximation factor is $2.06$~\cite{CMSY}}.
\end{lemma}
\begin{proof}
Since the reduced correlation clustering instance is a weighted correlation clustering instance, there exists an $O(\log n)$-approximation~\cite{demaine2006correlation}. Now since the weight of the edge between $p_i$ and $p_j$ in $G^\CP$ is at least $|P_i|\cdot|P_j|/2$ and at most $|P_i|\cdot|P_j|$, any two edges weights are within $2r^2$ of each other.  So if we remove the weights from $G^\CP$ and solve the resulting unweighted instance, we will get a $(2\rho r^2)$-approximation.
\end{proof}
\section{Fairlet Decomposition} 
\label{sec:decomposition}
In this section we show how to solve the fairlet decomposition problem by reducing it to a fair clustering problem with the median cost function in an appropriate metric space. The reduction (Section~\ref{sec:median}), loses a factor that is proportional to the size of the largest fairlet, but as we show in Section~\ref{sec:algorithms}, in cases that we know how to solve the fairlet decomposition problem, the size of the fairlets can be guaranteed to be small.
\subsection{Reduction to median cost}
\label{sec:median}
Consider a correlation clustering instance $G$ and let $d$ be a distance function defined on a metric space $M$ containing the set of vertices $V$. For a fairlet decomposition $\CP = \{P_1, P_2, \ldots\}$, we define the following median cost:
$\mcost(P_i) = \min_{u\in M} \sum_{v\in P_i} d(u,v)$ and $\mcost(\CP) = \sum_{P_i\in\CP}\mcost(P_i)$. Notice that the problem of finding the fairlet decomposition with minimum $\mcost(\CP)$ is precisely the fairlet decomposition problem for fair $k$-median, as studied by~\cite{chierichetti2017fair,bera2019fair}.

We now define a metric space $(M,d)$ such that the median cost $\mcost$ can approximate the fairlet cost $\fcost$. We first define an embedding $\phi: V \rightarrow [0,1]^n$ as follows.  For a vertex $u \in V$, let 
\[
\phi(u)_v = \left\{
\begin{array}{cl}
1 & \mbox{ if } u = v \mbox{ or } (u, v) \in E^+ \\
0 & \mbox{ if } (u, v) \in E^-.
\end{array}
\right.
\]
In other words, $\phi(v)$ is the $v$th row of the adjacency matrix of $G(V,E^+)$ after adding a positive self-loop at every vertex. Now we let $M=[0,1]^n$, and place the vertices in $V$ in this space using the mapping $\phi$. We let $d(\cdot, \cdot)$ be the Hamming distance between the points in $M$. In other words, for vertices $u,v\in V$, we have $d(u, v) = |\phi(u) - \phi(v)|$.  Intuitively, $d(u,v)$ measures the ``cost'' of committing to put $u$ and $v$ in one cluster in the correlation clustering instance.

We now prove the following two lemmas, which show that the $\fcost$ of a fairlet decomposition is close to its $\mcost$  with respect to the metric $d$. 
\ifappendix The proofs of these lemmas are in the Supplementary Material.
\else
The proofs of these lemmas are presented in the Supplemental Material.
\fi

\begin{lemmafirst}{lem:medbound1}
For any fairlet decomposition $\CP$, we have 
\[
\mcost(\CP) \leq 2 \cdot \fcost(\CP).
\]
\end{lemmafirst}

\begin{lemmafirst}{lem:medbound2}
Let $\CP$ be any fairlet decomposition and let $f = \max_{P \in \CP} |P|$.  Then, 
\[
\fcost(\CP) \leq 2f \cdot \mcost(\CP).
\]
\end{lemmafirst}
Using the above lemmas, we have the following.

\begin{thm}
\label{thm:medianreduction}
Assume there is a $\gamma$-approximation algorithm for fairlet decomposition with median costs. Furthermore, assume that this algorithm always produces fairlets of size at most $f$. Then the solution produced by this algorithm is a $(4f\gamma)$-approximation to the problem of finding a fairlet decomposition with minimum $\fcost$.
\end{thm}

\subsection{Algorithms for fairlet decomposition}
\label{sec:algorithms}
In this section, we give algorithms for fairlet decomposition with median cost for several notions of fairness. Using Theorem~\ref{thm:medianreduction}, these algorithms imply algorithms for fairlet decomposition problem and provide the algorithm $A_1$ in Theorems~\ref{thm:fairletreduction}. We focus on three fairness constraints: an upper bound of $\alpha=\frac12$ on the fraction of vertices of each color in each cluster; an upper bound of $\alpha = 1/C$ where $C$ is the number of distinct colors (this corresponds to the proportional fairness property studied in~\cite{chierichetti2017fair}); and an upper bound of $\alpha=1/t$ for an integer $t$ on the fraction of vertices of each color in each cluster. We give approximation algorithms for the first two cases and a bicriteria approximation for the third case, with upper bounds of $3$, $C$, and $2t-1$, respectively, on the size of fairlets. 

Throughout this subsection, when we speak of the cost of a fairlet decomposition, we mean its median cost.

\subsubsection{$\alpha = 1/2$}
\label{sec:fair-one-half-decomp}

This is probably the most common case of fair decomposition where clusters are required to not have a dominant color. In this case, we can show that a fairlets have size at most $3$ and find these fairlets by solving a minimum weight $2$-factor problem in a graph. Recall that a $2$-factor is a subgraph where each vertex has degree $2$ and edges may be used multiple times. This problem can be solved polynomially~\cite[Chapter 21]{schrijver2003combinatorial}. Define a graph $H$ on points in $V$ as follows: two vertices $u,v$ are connected by an edge if they have distinct colors; the weight of the edge is $d(u,v)$.  We first bound the cost of the optimal $2$-factor and then explain our approximation factor in the lemma. 
\begin{lemma}\label{lem:2factor}
The cost of an optimal $2$-factor in $H$ can be bounded by $2\cdot \mcost(\CP^*)$, where $\CP^*$ is the optimal fairlet decomposition.
\end{lemma}
\begin{proof}
We construct a feasible $2$-factor by constructing a $2$-factor for each fairlet $P\in \CP^*$ with center $\mu$. Since there are at most $|P|/2$ vertices of any color, depending on the parity of $P$, vertices of $P$ can be covered by matching and a possible multi-color triangle in $H$. Doubling the matching edges, we can get a $2$-factor for covering $\CP^*$. It remains to bound the cost of this $2$-factor. For a matching edge $(u,v)$ for $u,v\in P$, by triangle inequality, $d(u,v) \leq d(u,\mu_i) + d (v,\mu_i)$, and for a triangle $(u,v,w)$, the sum of pairwise distances can be bounded by $2(d(u,\mu) + d(v,\mu) + d(w,\mu))$. Hence the cost of the proposed $2$-factor for covering $\CP^*$ is at most $2 \cdot \mcost(\CP^*)$ and the overall cost of the optimal 2-factor is at most $2\cdot \mcost(\CP^*)$. 
\end{proof}

\begin{lemma}\label{lem:fair-decom-half}
For $\alpha = 1/2$, there is an approximation algorithm for fairlet decomposition that returns a solution with 
 \setlist{nolistsep}
    \begin{itemize}[noitemsep]
    \item median cost at most $2 \cdot \mcost(\CP^*)$.
    \item the size of largest fairlet is at most $3$.
    \item the size of smallest fairlet is at least $2$.
\end{itemize}
\end{lemma} 
\begin{proof}
Consider an optimal 2-factor in $H$. Define a fairlet decomposition as follows. For each cycle of even length, consider a set of alternating edges and let each alternating edge be a fairlet with one of the endpoints chosen as the center. For a cycle of odd length, there must exists three consecutive vertices of pairwise distinct colors. In this case, let one fairlet be these three vertices with the middle vertex chosen as the center and for the (unique) alternating edges covering the remaining vertices, let each edge be a fairlet with one of the endpoints chosen as the center. The median cost of these fairlets is at most the weight of the original 2-factor, which is at most $2 \cdot \mcost(\CP^*)$ by Lemma~\ref{lem:2factor}; the proof follows by construction. 
\end{proof}
Lemma~\ref{lem:fair-decom-half} and Lemma~\ref{lem:unconstraited-lem} yield the following.
\begin{thmfirst}{thm:dummy1}
For $\alpha = 1/2$, there is a $256$-approximation algorithm for fair correlation clustering.
\end{thmfirst}

\subsubsection{$\alpha = 1/C$ for $C$ colors}
\label{sec:fair-equal-decomposition}
In the case of $\alpha = 1/C$, each fairlet has equal number of points of each color and so these points can be matched together. We use this observation and devise an algorithm based on solving repeated matching problems in graph $H$ (construction explained in Section~\ref{sec:overview}).  
\begin{lemma}\label{lem:fair-decom-eqaul}
For $\alpha = 1/C$, there is an approximation algorithm for fairlet decomposition that returns a solution with 
 \setlist{nolistsep}
    \begin{itemize}[noitemsep]
    \item median cost at most $C \cdot \mcost(\CP^*)$,
    \item the size of each fairlet is $C$.
\end{itemize}
\end{lemma}
\begin{proof}
Consider an arbitrary ordering of the colors and solve a mincost matching problem between points of color $c$ and $c+1$ in the graph $H$. The union of these matchings yields a partition of $V$ into paths of length $C$. Each such path is a fairlet; let $\CP$ denote this fairlet decomposition. It remains to bound the cost of $\CP$.

Let us fix two colors $c$ and $c+1$ and let $M$ be an arbitrary matching between vertices of color $c$ and $c+1$ such that point $u$ is matched to a point $v$ only if $u$ and $v$ belong to the same partition of $\CP^*$. Since each part in $\CP^*$ has equal number of vertices of each color and there is an edge between any two vertices of different colors in $H$, matching $M$ exists. Now since the cost of each matching edge $(u,v)$ can be bounded by $d(u,\mu) + d(v,\mu)$ where $\mu$ is the center of the partition containing $u$ and $v$, the cost of $M$ can be bounded by median cost of serving clients of colors $c$ and $c+1$. Since each color is matched twice, the total cost of each path corresponding to a partition is at most $2 \cdot \mcost(\CP^*)$. Now for each path we pick the middle vertex as center and the cost of assigning vertices of the path to the center is at most $C/2$ cost of the path as each edge is charged at most $C/2$ times. Hence  $\mcost(\CP) \leq C \cdot \mcost(\CP^*)$.
\end{proof}

So in this case, we get a 2-approximation with fairlets of size at most $C$. Note that in the motivating application of fair clustering, the color of each vertex corresponds to one possible value of a sensitive feature like race or gender, and therefore the value $C$ tends to be small in such applications.

Lemma~\ref{lem:fair-decom-eqaul} and Lemma~\ref{lem:unconstraited-lem} yield the following.
\begin{thmfirst}{thm:dummy2}
For $\alpha =1/C$, there is a $(16.48 C^2)$-approximation algorithm for fair correlation clustering.
\end{thmfirst}

\subsubsection{$\alpha = 1/t$}\label{sec:fair-1-over-t}

The fair decomposition with median cost is not well studied in the literature and in this paper, we were able to devise algorithms for the case of $\alpha = 1/2$ and $\alpha = 1/C$.  Next we consider the case where $1/\alpha \in \ZZP$ and we argue how to utilize any approximation algorithm for fairlet decomposition as a black-box to build an algorithm for fair correlation clustering.  While we allow the black-box algorithm to produce fairlets of arbitrary size, the  following lemma ensures that the size of the fairlets can be bounded. 

\begin{lemma}\label{lem:fairletpartition}
For any set $P$ that satisfies fairness constraint with $\alpha = 1/t$, there exists a partition of $P$ into sets $(P_1, P_2, \ldots)$ where each $P_i$ satisfies the fairness constraint and $t \leq |P_i| < 2t$. 
\end{lemma}
\begin{proof}
Let $p = m \times t + r$  with $ 0 \leq r  < t$.  Then, the fairness constraints ensures that there are at most $m$ elements of each color. Consider the partitioning obtained through the following process: consider an ordering of elements where points of the same color are in consecutive places, assign points to sets $P_1, \ldots, P_m$ in a round-robin fashion. So each set $P_i$ gets at least $t$ elements and at most $t + r < 2t $ elements assigned to it. Since there are at most $m$ elements of each color, each set gets at most one point of any color and hence all sets satisfy the fairness constraint as $1 \leq \frac{1}{t} \cdot |P_i|$. 
\end{proof}

\begin{thmfirst}{thm:dummy3}
For $\alpha = 1/t$, given an $\gamma$-approximation algorithm for fairlet decomposition with median cost, there is an $O(t\gamma)$-approximation algorithm for fair correlation clustering.  
\end{thmfirst}

\newcommand{\amazon}{\textsf{amazon}\xspace}
\newcommand{\victorian}{\textsf{victorian}\xspace}
\newcommand{\reuters}{\textsf{reuters}\xspace}

\section{Experiments}

\begin{table*}
\small
\centering
\begin{tabular}{l|rr|rr|rrr}
\toprule
        Dataset &  \multicolumn{2}{c|}{\makecell{Unfair Alg.\\ \error}}  & \multicolumn{2}{c|}{\makecell{Unfair Alg.\\ \imbalancecolor}}  & \multicolumn{3}{c}{\makecell{Fair Alg. \error}}    \\
         &   \localsearch & \pivot   &   \localsearch & \pivot & \makecell{\match + \\ \localsearch} & \single & \randbaseline \\
\midrule
\amazon&0.010&0.011&0.40&0.39&0.064&0.786&0.215\\
\reuters, $\theta=0.25$&0.096& 0.161 &0.64&0.59&0.230&0.754&0.255\\
\reuters, $\theta=0.50$&0.181&0.231 &0.50&0.40&0.350&0.504&0.502\\
\reuters, $\theta=0.75$&0.188&0.241&0.15&0.25&0.199&0.252&0.746\\
\victorian, $\theta=0.25$&0.109&0.158&0.53& 0.46&0.212&0.753&0.251\\
\victorian, $\theta=0.50$&0.183&  0.268&0.31&0.23&0.348&0.502&0.499\\
\victorian, $\theta=0.75$&0.203&0.280&0.12&0.12&0.237&0.251&0.747\\
\midrule
mean over datasets&0.139&0.193&0.38&0.35&0.234&0.459&0.543\\
\bottomrule
\end{tabular}
\caption{\error and \imbalancecolor in $C=2$ color case for various datasets and different threshold $\theta$ for the quantile used for positive edges. Notice how our algorithm \match + \localsearch has cost comparable to \pivot and not much higher than \localsearch while reducing the imbalance from the up $65\%$ of the unfair algorithms to 0.} 
	\label{table:results-2-colors}
\end{table*}

In this section we present our experiments demonstrating that our algorithm solves the correlation clustering problem with fairness constraints with only a limited loss in the cost when compared to the vanilla (unfair) solution. We describe the datasets used, the algorithms evaluated, the quality measures, and our results. 

\para{Datasets}
We use publicly-available datasets from the UCI Repository\footnote{\url{archive.ics.uci.edu/ml}} and from the SNAP Datasets\footnote{\url{snap.stanford.edu/data/}}.

The datasets represent complete signed graphs from different domains, including both co-purchasing relationships among products, and semantic similarities among texts learned with embedding methods. The graphs are represented by complete signed matrices up to 1600x1600 in size, up to $0.9$ million positive edges, and up to $C=16$ colors. 
\ifappendix
Here we only briefly describe the datasets; for details see
Supplementary Material.

\else 
Here we only briefly describe the datasets, more details are in the appendix of the supplemental material.
\fi
\noindent {\bf \amazon}:
Vertices represents products on the Amazon website~\cite{leskovec2007dynamics}, the color is the item category, and two co-reviewed items  have a $+1$ weight edges (all non-co-reviewed items  have $-1$ weight edges). We use 1000 vertices equally distributed in two popular categories.

\noindent {\bf \reuters} and {\bf \victorian}:
These are extracted from text data used in previous fair clustering work~\cite{kdd19}.
The datasets include between $50$ and $100$ English language texts from each of up to 16 authors.\footnote{The datasets are available at \url{ archive.ics.uci.edu/ml/datasets/Reuter_50_50} and \url{archive.ics.uci.edu/ml/datasets/Victorian+Era+Authorship+Attribution}} Each vertex represents a text and the color represents the author. For each text we obtain a semantic embedding vector with standard methods. We use a threshold on the dot product of the embedding vectors to obtain the edges. Through this operation, we set the top $\theta \in \{0.25, 0.50, 0.75\}$ fraction of edges via dot products as $+1$'s, and the remaining edges are labeled $-1$.

\para{Algorithms}
We evaluate Algorithm~\ref{alg:main} in two fairness scenarios: with an upper bound of $\alpha = 1/2$  of the vertices for each color, and with $\alpha = 1/C$ for equal color representation (for the two-color case, the two are equivalent). For the $\alpha = 1/2$ case, in our experiments we simplify the algorithm of Section~\ref{sec:fair-equal-decomposition} to compute a minimum-cost perfect matching (1-factor)  instead of a 2-factor decomposition. This can be formally shown to be sufficient for the $C=2$ case, or when all optimal clusters are even sized, and we observe it works well in practice in our  experiments.
For $\alpha = 1/C$, we implement the repeated matching algorithm in~\ref{sec:fair-equal-decomposition} to obtain the fairlets in a similar fashion. After finding the fairlets, we use an in-house correlation clustering solver based on local search (\localsearch). We refer to our algorithms as \match + \localsearch for the $\alpha=1/2$ case and as \repmatch + \localsearch for the $\alpha=1/C$ case.

We also consider the following (unfair) baseline algorithms: the standard \pivot algorithm of Ailon et al.~\cite{acn} for correlation clustering (for \pivot, we repeat the randomized algorithm 10 times and use the best result), and the (unfair) local search heuristic \localsearch used as part of our algorithm. In addition, as no prior work has addressed the fair correlation clustering problem, we compare our algorithm with two simple fair baselines: the whole graph as one cluster \single, and a random fairlet decomposition \randbaseline. 

\begin{table*}
\small
\centering
\begin{tabular}{r|ccc}
\hline
Algorithm & \error & \imbalancecolor & \imbalancecolor \\
& & for 1/2 & for equality \\
\hline
\localsearch &0.249&0.011&0.218\\
\pivot & 0.345 &0.008 & 0.191 \\
\match + \localsearch &0.255	& 0	&0.180 \\
\repmatch + \localsearch &0.321& 0	& 0 \\
\single & 0.5 & 0 & 0 \\
\randbaseline	&0.5 & 0 & 0 \\
\hline
\end{tabular}
\caption{Experimental results for \victorian, $\theta = 0.50$, using $C=8$ colors.~\label{tab:results-victorian-8}}
\end{table*}

\para{Quality measures}
For each algorithm, we report the following measures. The \error of the correlation cost of the clustering obtained by the algorithm, presented as the ratio of the edges of the graph that are in disagreement with the clustering  (i.e.,  inter-cluster positive edges and  intra-cluster negative edges) over the number of edges (i.e., 0 corresponds to a perfect solution, whereas 1 corresponds to a completely incorrect solution). For fairness, we report the \imbalancecolor as the total fraction of vertices that violate the $\alpha$ color representation constraint, i.e., vertices for each cluster and color that are above the $\alpha$ fraction for the size of the cluster~\cite{kdd19}. More precisely, let $P$ be a cluster in the solution of an $\alpha$-color constraint instance. The maximum allowed number of points of a certain color in the cluster $P$ is $\lfloor |P|\alpha\rfloor$. Let $V_c$ be the vertices of color $c$ on the graph and $\Delta_P = \sum_{P,c} \max(|P \cap V_c| - \lfloor |P|\alpha\rfloor, 0)$ be the vertices violating the constraint in $P$. We report $\sum_P \Delta_P / |V|$ as the \imbalancecolor for $\alpha \in \{1/2, 1/C\}$ (i.e., 0 corresponds to no imbalance,  and 1 corresponds to complete imbalance). We repeat all algorithms 10 times and report the mean results for all measures. 

\subsection{Results for $C=2$}
Table~\ref{table:results-2-colors} shows the results for the $C=2$ color case with $\alpha=1/2$, i.e., equal representation over the various datasets.  We use the \match+ \localsearch as a fair algorithm, which has an \imbalancecolor of 0 by construction.

The table shows clearly that  our algorithm \match + \localsearch obtains clusters that are fair and has costs comparable to the unfair \pivot baseline (and sometimes better) and slightly worse than the \localsearch baseline. On average, over all datasets, our algorithm has an average cost of $0.234$ vs $0.193$ for \pivot and $0.139$ for \localsearch. On the other hand, our algorithm is significantly better than both 
\single and \randbaseline, which, while satisfying fairness, have exorbitantly high costs. 

Notice how the \localsearch and the \pivot baselines have very high \imbalancecolor values of up to $65\%$ of the vertices (as they are oblivious to colors), showing the importance of developing novel algorithms for the problem. This result is not surprising. If pairs of vertices of the same color are more likely to be similar, it is expected that many clusters will contain vast majorities of points with a single color.

\subsection{Results for $C>2$}
Here, we study the behavior of our algorithm in the case when more than two colors are present in the dataset. We use \match+ \localsearch to obtain a $\alpha=1/2$ fair solution and  \repmatch+ \localsearch  to obtain a $\alpha=1/C$ (equal representation) solution.

We report an overview of our experimental results in Table~\ref{tab:results-victorian-8} for the dataset \victorian with threshold $\theta = 0.50$ and $C=8$ colors (note that this is a different graph than that produced with the previous $C=2$ dataset).
\ifappendix 
More experimental results are available in 
Supplementary Material.
\else
More experimental results are available in the extended material.
\fi

It is easy to see that all the earlier trends continue to hold.  Notice how the algorithm for the $1/2$ case \match+ \localsearch is only marginally worse than the best unfair solution \localsearch and much better than all other baselines. Our algorithm for the more difficult equal representation case is again better than the \pivot baseline. Notice how all unfair algorithms are significantly far from being equally balanced. However, the presence of many colors makes it easier to get closer to the $1/2$ threshold.

\begin{figure}
\centering
\includegraphics[width=0.30\textwidth,keepaspectratio]{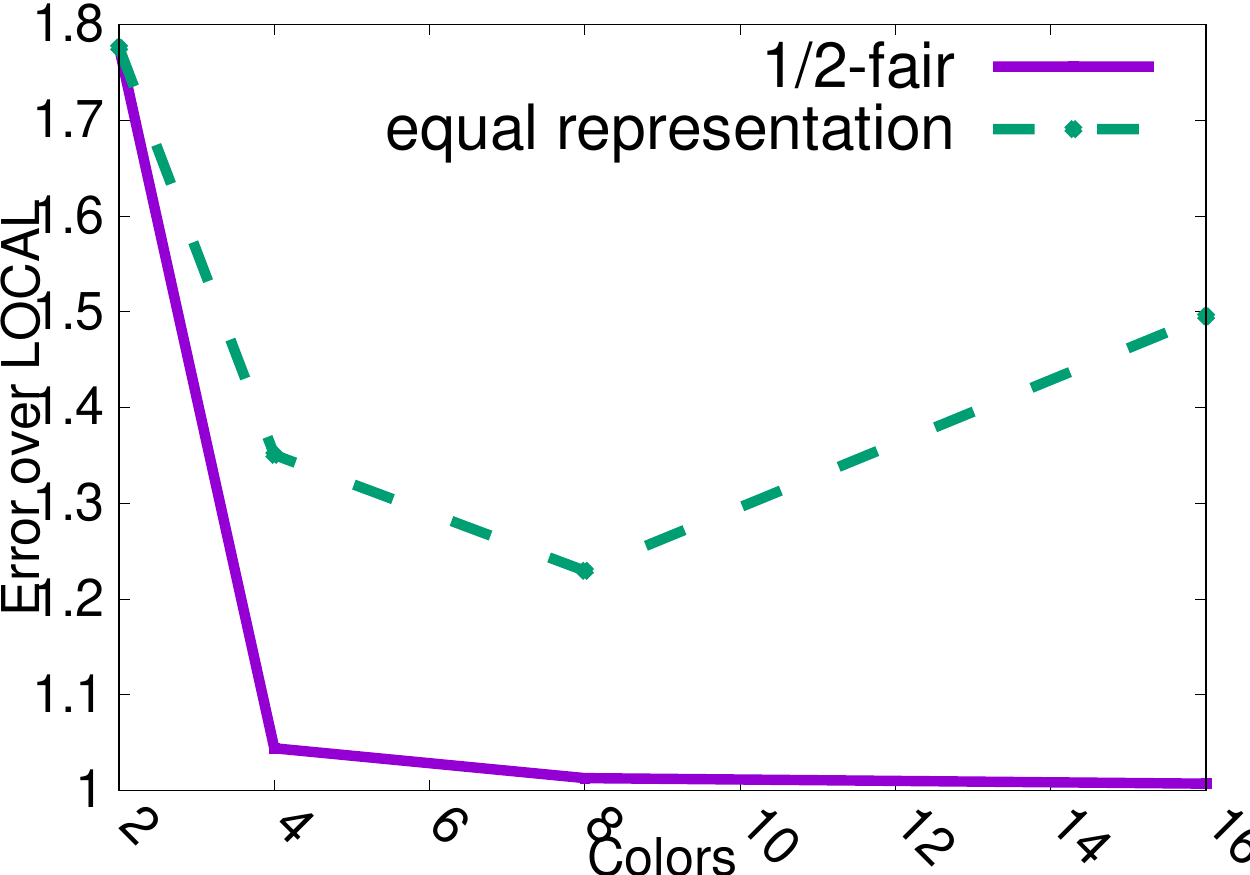}
\caption{\error of our algorithms over that of the unfair \localsearch algorithms for $\alpha=1/2$ and $\alpha=1/C$, on a series of graphs from \victorian, $\theta = 0.50$, and using $C=2$ to $C=16$ colors. \label{fig:result-color}}.
\end{figure}

We confirm this observation in Figure~\ref{fig:result-color}. The figure shows a comparison of the \error for our algorithms for $\alpha=1/2$ and $\alpha=1/C$ vs  the \error for the vanilla \localsearch algorithms as the number of colors goes from $C=2$ to $C=16$. We report the ratio of \match + \localsearch over \localsearch as a solid line, and the ratio for \repmatch + \localsearch as a dashed line. Notice how the error for our algorithm for the $\alpha=1/2$ case gets closer and closer to the \localsearch output error for more colors. Again, this result is obtained because the presence of many colors makes the problem easier for the $\alpha=1/2$ case. The performance of the algorithm for the $\alpha=1/C$ case has a less stable pattern, but it confirms that the algorithm is quite competitive with the unfair solution (between $30-80\%$ higher error) even for quite a few colors. These results are significantly better than what is expected from a worst-case analysis.  

Finally, we report having observed that using just \match or \repmatch fairlets without the re-clustering part of the algorithm is not sufficient for obtaining good results, as the re-clustering step is needed for obtaining clusters that do not have large errors.


\section{Conclusions}
In this paper we initiated the study of correlation clustering with fairness constraints. We showed a reduction to the fairlet decomposition problem with a standard median cost function, for a carefully chosen distance function. Using this, and old and new results on the fairlet decomposition problem with a median cost function, we obtained provable constant-factor approximation algorithms for fair correlation clustering for various notions of fairness. Our experimental evaluation shows that these algorithms perform well not only in theory but also in practice.

\balance
\bibliography{paper}

\begin{thebibliography}{10}

\bibitem{kdd19}
S.~Ahmadian, A.~Epasto, R.~Kumar, and M.~Mahdian.
\newblock Clustering without over-representation.
\newblock In {\em KDD}, pages 267--275, 2019.

\bibitem{DBLP:conf/focs/AhmadianNSW17}
S.~Ahmadian, A.~Norouzi{-}Fard, O.~Svensson, and J.~Ward.
\newblock Better guarantees for $k$-means and euclidean $k$-median by
  primal-dual algorithms.
\newblock In {\em FOCS}, pages 61--72, 2017.

\bibitem{ahn2015correlation}
K.~Ahn, G.~Cormode, S.~Guha, A.~McGregor, and A.~Wirth.
\newblock Correlation clustering in data streams.
\newblock In {\em ICML}, pages 2237--2246, 2015.

\bibitem{acn}
N.~Ailon, M.~Charikar, and A.~Newman.
\newblock Aggregating inconsistent information: Ranking and clustering.
\newblock {\em J. {ACM}}, 55(5):23:1--23:27, 2008.

\bibitem{scalablefair}
A.~Backurs, P.~Indyk, K.~Onak, B.~Schieber, A.~Vakilian, and T.~Wagner.
\newblock Scalable fair clustering.
\newblock In {\em ICML}, pages 405--413, 2019.

\bibitem{Bansal2004}
N.~Bansal, A.~Blum, and S.~Chawla.
\newblock Correlation clustering.
\newblock {\em Mach. Learn.}, 56(1-3):89--113, 2004.

\bibitem{bera2019fair}
S.~K. Bera, D.~Chakrabarty, N.~Flores, and M.~Negahbani.
\newblock Fair algorithms for clustering.
\newblock In {\em NeurIPS}, pages 4955--4966, 2019.

\bibitem{bercea2018fair}
I.~O. Bercea, M.~Gross, S.~Khuller, A.~Kumar, C.~Rösner, D.~R. Schmidt, and
  M.~Schmidt.
\newblock On the cost of essentially fair clusterings.
\newblock In {\em APPROX-RANDOM}, pages 18:1--18:22, 2019.

\bibitem{bressan2019correlation}
M.~Bressan, N.~Cesa-Bianchi, A.~Paudice, and F.~Vitale.
\newblock Correlation clustering with adaptive similarity queries.
\newblock In {\em NeurIPS}, pages 12510--12519, 2019.

\bibitem{calders2010three}
T.~Calders and S.~Verwer.
\newblock Three naive bayes approaches for discrimination-free classification.
\newblock {\em DMKD}, 21(2):277--292, 2010.

\bibitem{celis2018multiwinner}
L.~E. Celis, L.~Huang, and N.~K. Vishnoi.
\newblock Multiwinner voting with fairness constraints.
\newblock In {\em IJCAI}, pages 144--151, 2018.

\bibitem{celis2017ranking}
L.~E. Celis, D.~Straszak, and N.~K. Vishnoi.
\newblock Ranking with fairness constraints.
\newblock In {\em ICALP}, pages 28:1--28:15, 2018.

\bibitem{charikar2005clustering}
M.~Charikar, V.~Guruswami, and A.~Wirth.
\newblock Clustering with qualitative information.
\newblock {\em JCSS}, 71(3):360--383, 2005.

\bibitem{CMSY}
S.~Chawla, K.~Makarychev, T.~Schramm, and G.~Yaroslavtsev.
\newblock Near optimal {LP} rounding algorithm for correlation clustering on
  complete and complete $k$-partite graphs.
\newblock In {\em STOC}, pages 291--228, 2015.

\bibitem{chen2019proportionally}
X.~Chen, B.~Fain, C.~Lyu, and K.~Munagala.
\newblock Proportionally fair clustering.
\newblock In {\em ICML}, pages 1032--1041, 2019.

\bibitem{chierichetti2017fair}
F.~Chierichetti, R.~Kumar, S.~Lattanzi, and S.~Vassilvitskii.
\newblock Fair clustering through fairlets.
\newblock In {\em NIPS}, pages 5029--5037, 2017.

\bibitem{matroidfair}
F.~Chierichetti, R.~Kumar, S.~Lattanzi, and S.~Vassilvitskii.
\newblock Matroids, matchings, and fairness.
\newblock In {\em AISTATS}, pages 2212--2220, 2019.

\bibitem{dash2019summarizing}
A.~Dash, A.~Shandilya, A.~Biswas, K.~Ghosh, S.~Ghosh, and A.~Chakraborty.
\newblock Summarizing user-generated textual content: Motivation and methods
  for fairness in algorithmic summaries.
\newblock In {\em CSCW}, pages 1--28, 2019.

\bibitem{demaine2006correlation}
E.~D. Demaine, D.~Emanuel, A.~Fiat, and N.~Immorlica.
\newblock Correlation clustering in general weighted graphs.
\newblock {\em TCS}, 361(2-3):172--187, 2006.

\bibitem{dwork2012fairness}
C.~Dwork, M.~Hardt, T.~Pitassi, O.~Reingold, and R.~Zemel.
\newblock Fairness through awareness.
\newblock In {\em ITCS}, pages 214--226, 2012.

\bibitem{elzayn2019fair}
H.~Elzayn, S.~Jabbari, C.~Jung, M.~Kearns, S.~Neel, A.~Roth, and Z.~Schutzman.
\newblock Fair algorithms for learning in allocation problems.
\newblock In {\em FAT}, pages 170--179, 2019.

\bibitem{fang2019achieving}
B.~Fang, M.~Jiang, and J.~Shen.
\newblock Achieving fairness in determining medicaid eligibility through
  fairgroup construction.
\newblock {\em arXiv:1906.00128}, 2019.

\bibitem{feldman2015certifying}
M.~Feldman, S.~A. Friedler, J.~Moeller, C.~Scheidegger, and
  S.~Venkatasubramanian.
\newblock Certifying and removing disparate impact.
\newblock In {\em KDD}, pages 259--268, 2015.

\bibitem{fish2016confidence}
B.~Fish, J.~Kun, and {\'A}.~D. Lelkes.
\newblock A confidence-based approach for balancing fairness and accuracy.
\newblock In {\em SDM}, pages 144--152, 2016.

\bibitem{gonzalez1985clustering}
T.~F. Gonzalez.
\newblock Clustering to minimize the maximum intercluster distance.
\newblock {\em TCS}, 38:293--306, 1985.

\bibitem{gungor2018fifty}
A.~Gungor.
\newblock Fifty victorian era novelists authorship attribution data.
\newblock 2018.

\bibitem{har2019near}
S.~Har-Peled and S.~Mahabadi.
\newblock Near neighbor: Who is the fairest of them all?
\newblock {\em arXiv:1906.02640}, 2019.

\bibitem{hochbaum1985best}
D.~S. Hochbaum and D.~B. Shmoys.
\newblock A best possible heuristic for the $k$-center problem.
\newblock {\em MOR}, 10(2):180--184, 1985.

\bibitem{huang2019coresets}
L.~Huang, S.~H.~C. Jiang, and N.~K. Vishnoi.
\newblock Coresets for clustering with fairness constraints.
\newblock In {\em NeurIPS}, pages 7587--7598, 2019.

\bibitem{jain2010data}
A.~K. Jain.
\newblock Data clustering: 50 years beyond $k$-means.
\newblock {\em Pattern Recognition Letters}, 31(8):651--666, 2010.

\bibitem{joseph2016fairness}
M.~Joseph, M.~Kearns, J.~H. Morgenstern, and A.~Roth.
\newblock Fairness in learning: Classic and contextual bandits.
\newblock In {\em NIPS}, pages 325--333, 2016.

\bibitem{kalhan2019improved}
S.~Kalhan, K.~Makarychev, and T.~Zhou.
\newblock Improved algorithms for correlation clustering with local objectives.
\newblock In {\em NeurIPS}, pages 9341--9350, 2019.

\bibitem{kamishima2012fairness}
T.~Kamishima, S.~Akaho, H.~Asoh, and J.~Sakuma.
\newblock Fairness-aware classifier with prejudice remover regularizer.
\newblock In {\em PKDD}, pages 35--50, 2012.

\bibitem{kamishima2011fairness}
T.~Kamishima, S.~Akaho, and J.~Sakuma.
\newblock Fairness-aware learning through regularization approach.
\newblock In {\em ICDMW}, pages 643--650, 2011.

\bibitem{kanungo2004local}
T.~Kanungo, D.~M. Mount, N.~S. Netanyahu, C.~D. Piatko, R.~Silverman, and A.~Y.
  Wu.
\newblock A local search approximation algorithm for $k$-means clustering.
\newblock {\em Computational Geometry}, 28(2-3):89--112, 2004.

\bibitem{kcenterfair}
M.~Kleindessner, P.~Awasthi, and J.~Morgenstern.
\newblock Fair $k$-center clustering for data summarization.
\newblock In {\em ICML}, pages 3458--3467, 2019.

\bibitem{spectralfair}
M.~Kleindessner, S.~Samadi, P.~Awasthi, and J.~Morgenstern.
\newblock Guarantees for spectral clustering with fairness constraints.
\newblock In {\em ICML}, pages 3448--3457, 2019.

\bibitem{kushagra2018semi}
S.~Kushagra, S.~Ben-David, and I.~Ilyas.
\newblock Semi-supervised clustering for de-duplication.
\newblock {\em arXiv:1810.04361}, 2018.

\bibitem{leskovec2007dynamics}
J.~Leskovec, L.~A. Adamic, and B.~A. Huberman.
\newblock The dynamics of viral marketing.
\newblock {\em TWEB}, 1(1):5, 2007.

\bibitem{liyizhang}
J.~Li, K.~Yi, and Q.~Zhang.
\newblock Clustering with diversity.
\newblock In {\em ICALP}, pages 188--200, 2010.

\bibitem{DBLP:journals/siamcomp/LiS16}
S.~Li and O.~Svensson.
\newblock Approximating $k$-median via pseudo-approximation.
\newblock {\em SICOMP}, 45(2):530--547, 2016.

\bibitem{pan2015parallel}
X.~Pan, D.~Papailiopoulos, S.~Oymak, B.~Recht, K.~Ramachandran, and M.~I.
  Jordan.
\newblock Parallel correlation clustering on big graphs.
\newblock In {\em NIPS}, pages 82--90, 2015.

\bibitem{pouget2019variance}
J.~Pouget-Abadie, K.~Aydin, W.~Schudy, K.~Brodersen, and V.~Mirrokni.
\newblock Variance reduction in bipartite experiments through correlation
  clustering.
\newblock In {\em NeurIPS}, pages 13288--13298, 2019.

\bibitem{rosner2018privacy}
C.~R{\"o}sner and M.~Schmidt.
\newblock Privacy preserving clustering with constraints.
\newblock In {\em ICALP}, pages 96:1--96:14, 2018.

\bibitem{schrijver2003combinatorial}
A.~Schrijver.
\newblock {\em Combinatorial Optimization: Polyhedra and Efficiency},
  volume~24.
\newblock Springer Science \& Business Media, 2003.

\bibitem{yang2017measuring}
K.~Yang and J.~Stoyanovich.
\newblock Measuring fairness in ranked outputs.
\newblock In {\em SSDBM}, pages 22:1--22:6, 2017.

\bibitem{ziko2019clustering}
I.~M. Ziko, E.~Granger, J.~Yuan, and I.~B. Ayed.
\newblock Clustering with fairness constraints: A flexible and scalable
  approach.
\newblock {\em arXiv:1906.08207}, 2019.

\end{thebibliography}
\bibliographystyle{abbrv}
\balance

\ifappendix

\newpage
\appendix
\section{Deferred Proofs of Section \ref{sec:overview}}
\label{sec:proofs1}

\begin{lemmaagain}{lem:1}
Given a correlation clustering instance $G$, a fairlet decomposition $\CP$ for $G$, and a clustering $\CC$ of $G$, there exists a clustering $\CC'$ of $G^\CP$ such that
\[
    \cost(G^\CP,\CC') \leq \cost(G,\CC) + \out(\CP).
\]
\end{lemmaagain}

\begin{proof}
To show existence of the claimed clustering $\CC'$, we devise a randomized algorithm and bound the expected cost of the clustering output of this algorithm. We abuse notation and let clustering $\CC'$ be the output of this randomized algorithm on $G^\CP$. Given fairlet decomposition $\CP$, the algorithm first picks a representative $r_i$ for each partition $P_i$ in $\CP$ uniformly at random. Next the algorithm defines clustering $\CC'$ based on where $r_i$ is assigned in $\CC$: if $r_i$ is placed in cluster $C_j$, it places the vertex $p_i$ of $G^\CP$ in the cluster $C'_j$. 

Next, we bound the expected cost of $\CC'$ in $G^\CP$. Let us fix two vertices $p_i$ and $p_j$ in $G^\CP$.  We first consider the case $\sigma(p_i, p_j) > 0$; the other case follows by a similar argument.  The clustering $\CC'$ incurs a cost of $|\sigma(p_i, p_j)|$ if $p_i$ and $p_j$ are assigned to different clusters, which happens when the two representatives $r_i$ and $r_j$ are in different clusters in $\CC$. Now, if $\sigma(r_i, r_j) = +1$, then  $\CC$ also pays a cost of $1$ for separating $r_i$ and $r_j$ and if $\sigma(r_i, r_j) = -1$, then the edge $(r_i, r_j)$ is an edge with the minority sign and it contributes to $\out(P_i, P_j)$.  Hence, if we denote the cost of the edge between $p_i$ and $p_j$ in the clustering $\CC'$ by $\cost_{\CC'}(p_ip_j)$, we can bound the expected value of this cost as follows:
 \newcommand{\ind}[1]{\mathds{1}\left(#1\right)}
\begin{eqnarray*}
\expct{\cost_{\CC'}(p_i, p_j)} 
& \leq & |\sigma(p_i, p_j)| \cdot E[\cost_{\CC}(r_i,r_j) \\
& & + \ind{\sigma(p_i,p_j) * \sigma(r_i,r_j) < 0}],
\end{eqnarray*}
where $\ind{A}$ is an indicator function having value 1 if $A$ is true and zero otherwise. Now since $r_i$ and $r_j$ are picked uniformly at random, 
\begin{eqnarray*}
\expct{\cost_{\CC'}(p_i, p_j)}
&\\ \leq  \frac{|\sigma(p_i, p_j)|}{\size{P_i} \cdot \size{P_j}} \cdot (\cost_{\CC}(P_i,P_j) + \out(P_i, P_j)). &
\end{eqnarray*}
This follows from the fact that there are $\size{P_i} \cdot \size{P_j}$ many possible pairs $(r_i,r_j)$ to be selected. Summing the above over all $p_i, p_j$ and using $|\sigma(p_i, p_j)| \leq |P_i|\cdot |P_j|$, we get
$$
\expct{\cost(G^\CP, \CC')} \leq \cost(\CC) + \out(\CP).
\qedhere
$$
\end{proof}

\begin{lemmaagain}{lem:2}
Assume $\CC$ is a clustering of $G^{\CP}$, and let $\CC'$ be the clustering computed in line 4 of Algorithm~\ref{alg:main} for $G$. Then we have
\[
\cost(G,\CC') \leq \cost(G^\CP,\CC) + \fcost(\CP).
\]
\end{lemmaagain}
\begin{proof}
Any edge $(u, v)$ contributing to the cost of the clustering $\CC'$ is either a negative edge inside a fairlet or an edge between two fairlets that are clustered in disagreement with $\sigma(u,v)$ in $\CC$. Negative edges inside fairlets are counted in $\inn(\CP)$. An edge $(u, v)$ between fairlets $P_i$ and $P_j$ that is clustered in disagreement with $\sigma(u,v)$ either has the same sign as the  majority sign of $E(P_i,P_j)$, or as the minority sign of $E(P_i,P_j)$. The edges in the former case are counted in  $\cost(G^\CP,\CC)$, and the edges in the latter case are counted in $\out(\CP)$. Therefore, the total cost of $\CC'$ is at most $\cost(G^\CP,\CC) + \inn(\CP) + \out(\CP)$.
\end{proof}

\begin{lemmaagain}{lem:3}
For any constrained correlation clustering instance $G$, and any constrained clustering $\CC$ of $G$, there is a fairlet decomposition $\CP$ of $G$ satisfying $\fcost(\CP)\le \cost(G,\CC)$.
\end{lemmaagain} 
\begin{proof}
We can simply take $\CP$ to be the same as the clustering $\CC$. It is easy to observe that this is a valid fairlet decomposition. To bound $\fcost(\CP)$, it is enough to note that each edge counted in $\inn(\CP)$ also imposes a cost of 1 in $\CC$ (as it is a negative edge inside a cluster), and for any two clusters $C_i$ and $C_j$, the number of positive edges between $C_i$ and $C_j$ is at least $\out(C_i,C_j)$. Summing over all these inequalities, we obtain that $\fcost(\CP)$ is at most $\cost(G,\CC)$.
\end{proof}

\section{Deferred Proofs of Section \ref{sec:decomposition}}
\label{sec:proofs2}

\begin{lemmaagain}{lem:medbound1}
For any fairlet decomposition $\CP$, we have 
\[
\mcost(\CP) \leq 2 \cdot \fcost(\CP).
\]
\end{lemmaagain}
\begin{proof}
Consider a fairlet $P_i$ in $\CP$. We define a vector $\mu\in[0,1]^n$ indexed by the vertices of $G$ as follows: $\mu_u = {\rm majority}(\{\phi(v)_u : v \in P_i\})$. By the definition of $\mcost$, we have
\begin{eqnarray}\label{eq:medcost}
\mcost(P_i)&=&\min_{u\in [0,1]^n}\sum_{v\in P_i} d(u,v) \nonumber\\
&\leq&\sum_{v \in P_i} |\mu - \phi(v)|. 
\end{eqnarray}
On the other hand, for every vertex $u$, if we denote $N^-(u) = \{v: (u, v) \in E^-\}$ and
$N^+(u) = \{u\}\cup\{v: (u, v) \in E^+\}$, we have
\begin{eqnarray*}
\lefteqn{
\sum_{v \in P_i} |\mu_u - \phi(v)_u|  
= |\{v\in P_i: \phi(v)_u \neq \mu_u|\}|} \\
&=& \min(|N^-(u)\cap P_i|, |N^+(u)\cap P_i|).
\end{eqnarray*}
For $u\not\in P_i$, the above quantity is precisely $\out(P_i,\{u\})$. For $u\in P_i$, the above quantity is at most $|N^-(u)\cap P_i|$, which is the number of negative edges in $P_i$ incident on $u$. Therefore, the sum of this quantity over all $u$ can be bounded by

\begin{eqnarray}
\label{eqn:medcost2}
\lefteqn{
\sum_{u\in V}\sum_{v \in P_i} |\mu_u - \phi(v)_u|} \\
& \le 2\cdot\inn(P_i)
+\sum_{u\in V\setminus P_i}\out(P_i,\{u\}).\nonumber
\end{eqnarray}

Finally, by the definition of $\out$, we have $\out(P_i,S)+\out(P_i,T)\le\out(P_i,S\cup T)$ for any two disjoint sets $S$ and $T$. Therefore, $$\sum_{u\in V\setminus P_i}\out(P_i,\{u\})\le \sum_{j\neq i} \out(P_i,P_j).$$ 
Combining this with Equations~\eqref{eq:medcost} and \eqref{eqn:medcost2}, we get $\mcost(\CP) \leq 2\cdot\inn(\CP) + \out(\CP)\le 2\cdot\fcost(\CP)$.
\end{proof}

\begin{lemmaagain}{lem:medbound2}
Let $\CP$ be any fairlet decomposition and let $f = \max_{P \in \CP} |P|$.  Then, 
\[
\fcost(\CP) \leq 2f \cdot \mcost(\CP).
\]
\end{lemmaagain}
\begin{proof}
Consider a fairlet $P_i$ and define the vector $\mu$ as in the proof of Lemma~\ref{lem:medbound1}. It is easy to see that
\begin{eqnarray*}
\mcost(P_i) &=&
\min_{x\in[0,1]^n}\sum_{v\in P_i}|x-\phi(v)|\\
&=&|\mu-\phi(v)|\\
&=&\sum_{u\in V}\min(|N^-(u)\cap P_i|, |N^+(u)\cap P_i|).
\end{eqnarray*}

As in the proof of Lemma~\ref{lem:medbound1}, for every $u\not\in P_i$, the summand in the above expression is precisely
$\out(P_i,\{u\})$. For $u\in P_i$, this quantity is zero if $u$ has no negative edge to any other vertex in $P_i$, and is at least 1 otherwise. Therefore, since $|P_i|\le f$, this quantity is always at least the number of negative edges from $u$ to other vertices in $P_i$ divided by $f$. Therefore,

\begin{eqnarray*}
\lefteqn{
\mcost(P_i)} \\
& \ge & \frac2f\cdot\inn(P_i) + \sum_{j:j\neq i}\sum_{u\in P_j}\out(P_i,\{u\}).
\end{eqnarray*}

Summing over all $i$ and rearranging the terms, we obtain:

\begin{eqnarray}
\label{eqn:mcostlb}
\mcost(\CP) &\ge& \frac2f\cdot\inn(\CP)\\
&&\,+\sum_{i<j}\left(\sum_{u\in P_j}\out(P_i,\{u\})\right.\nonumber\\
&&\hspace{15mm}\,\left.
+\sum_{u\in P_i}\out(P_j,\{u\})\right).\nonumber
\end{eqnarray}

Now, we fix $i<j$ and bound the summand in the above expression. Without loss of generality, we assume $|P_i|\ge |P_j|$. We consider the following cases:

\begin{description}
\item[Case 1: ] There are at most $|P_i|/2$ vertices $u$ in $P_i$ with $\out(P_j,\{u\}) = 0$. In this case, we have
$\sum_{u\in P_i}\out(P_j,\{u\}) \ge \frac{|P_i|}2\ge\frac{|P_i|\cdot|P_j|}{2f}\ge \frac{\out(P_i,P_j)}{2f}$.
\item[Case 2: ] There are at least $|P_i|/2$ vertices $u$ in $P_i$ with $\out(P_j,\{u\}) = 0$. Let $S$ be the set of such vertices. By the definition of $\out(P_j,\{u\})$, any $u\in S$ must have either positive edges to all vertices in $P_j$, or negative edges to all of them. Assume $x$ vertices in $S$ have positive edges to all vertices in $P_j$ and $y$ of them have negative edges, for some $x, y$ with $x+y=|S|\ge |P_i|/2$. We further consider the following cases:
\begin{description}
\item[Case 2a: ] If $x=0$, then every vertex $u$ in $P_j$ has at least $P_i/2$ positive edges to vertices in $P_i$ (namely, at least to those in $S$). Therefore, $\out(P_i,\{u\}) = |E^-\cap E(P_i,\{u\})|$. Thus, $\sum_{u\in P_j}\out(P_i,\{u\}) = 
|E^-\cap E(P_i,P_j)|\ge\out(P_i,P_j).$
\item[Case 2b: ] If $y=0$, an argument similar to case 2a shows that $\sum_{u\in P_j}\out(P_i,\{u\}) = 
|E^+\cap E(P_i,P_j)|\ge\out(P_i,P_j).$
\item[Case 2c: ] If $x \ge 1$ and $y\ge 1$, then each vertex in $P_j$ has at least one positive edge and at least one negative edge to $P_i$. Therefore, for every $u\in P_j$, $\out(P_i,\{u\})\ge 1$. Thus, $\sum_{u\in P_j}\out(P_i,\{u\})\ge |P_j| \ge \frac{|P_i|\cdot|P_j|}{f}\ge \frac{\out(P_i,P_j)}{f}$.
\end{description}
\end{description}
In all of the above cases, we have:
\begin{eqnarray*}
\displaystyle\sum_{u\in P_j}\out(P_i,\{u\})+\sum_{u\in P_i}\out(P_j,\{u\}) \\
\\ \hspace*{50mm} \ge \frac{\out(P_i,P_j)}{2f}.
\end{eqnarray*}
This, together with \eqref{eqn:mcostlb}, implies:
\begin{eqnarray*}
\mcost(\CP)&\ge& \frac2f\cdot\inn(\CP)+\frac1{2f}\cdot\out(\CP)
\\&\ge& \frac1{2f}\cdot\fcost(\CP).
\qedhere
\end{eqnarray*}
\end{proof}

\begin{thmagain}{thm:dummy1}
For $\alpha =1/2$, there is a $256$-approximation algorithm for fair correlation clustering.
\end{thmagain}
\begin{proof}
From Theorem~\ref{thm:medianreduction} and Lemma~\ref{lem:fair-decom-half}, we get a $24$-approximation algorithm for solving fairlet decomposition with minimum $\fcost$ ($f = 3, \gamma = 2$). From Lemma~\ref{lem:unconstraited-lem}, there is a $2\cdot 2.06 \cdot (1.5)^2 = 9.27$-approximation algorithm for unconstrained correlation clustering.  Combining, we get a $255.75$-approximation algorithm for fair correlation clustering. 
\end{proof}

\begin{thmagain}{thm:dummy2}
For $\alpha =1/C$, there is a $(16.48 C^2)$-approximation algorithm for fair correlation clustering.
\end{thmagain}
\begin{proof}
From Theorem~\ref{thm:medianreduction} and Lemma~\ref{lem:fair-decom-half}, we get a $4C^2$-approximation algorithm for solving fairlet decomposition with minimum $\fcost$ ($f = C, \gamma = C$).  From  Lemma~\ref{lem:unconstraited-lem}, there is a $2\cdot 2.06 \cdot 1 = 4.12$-approximation algorithm for unconstrained correlation clustering.  Combining, we get a $(16.48 C^2)$-approximation algorithm for fair correlation clustering. 
\end{proof}

\begin{thmagain}{thm:dummy3}
For $\alpha = 1/t$, given an $\gamma$-approximation for fair decomposition with median cost, there exists an  $O(t\gamma)$-approximation algorithm for fairlet correlation clustering.  
\end{thmagain}
\begin{proof}
Let $\CP$ be the output of the output of the $\gamma$-approximation algorithm on the metric space $(M,d)$ obtained from correlation instance $G$. Let fairlet decomposition $\CP'$ be obtained from $\CP$ by applying Lemma~\ref{lem:fairletpartition} and assigning each fairlet to a center minimizing the median cost of the fairlet. Since dedicating a center to a subset of points assigned to the same center in $\CP$ can only decrease the median cost, $\mcost(\CP') \leq \mcost(\CP)$. 
From Theorem~\ref{thm:medianreduction} and Lemma~\ref{lem:fair-decom-half}, there is a $((8t - 4) \gamma)$-approximation algorithm for solving fairlet decomposition with minimum $\fcost$ ($f = 2t - 1, \gamma = \gamma_\cA$). Since the size of each fairlet is at least $t$, applying Lemma~\ref{lem:unconstraited-lem}, there is a $2\cdot 2.06 \cdot 
(\frac{2t - 1} {t})^2 < 16.48$-approximation algorithm for solving unconstrained correlation clustering. Now applying Theorem~\ref{thm:fairletreduction}, we get an $O(t\gamma)$-approximation algorithm for fair correlation clustering.
\end{proof}

\section{Supplemental Experimental Results}

Here, we report additional experimental results.

\label{appendix:exp}
\subsection{Description of the datasets}
We describe more in detail the datasets used.

\noindent {\bf \amazon}:
Vertices represents products on the Amazon website~\cite{leskovec2007dynamics} and positive edges connect products co-reviewed by the same user (all missing edges are treated as negative). We set the color of each item to its category. Further, we use 1000 vertices equally distributed among 2 popular book categories {\it Nonfiction} and {\it Literature \& Fiction} for a total of $\sim 106{,}000$ positive edges.

\noindent {\bf \reuters}:
This graph is extracted from a dataset, which was used in previous fair clustering work~\cite{kdd19} and includes 50 English language articles from each of up to 16 authors (for a total of up to 800 texts).This dataset is available at \url{ archive.ics.uci.edu/ml/datasets/Reuter_50_50}. We transform each text into a 10-dimensional vector using Gensim’s Doc2Vec with standard parameters, as in previous work~\cite{kdd19}, and we create one vertex for each text. Then we use a threshold on the dot product of the embedding vectors. Through this operation, we set the top $\theta \in \{0.25, 0.50, 0.75\}$ fraction of edges via dot products as $+1$'s ,and the remaining edges are assigned $-1$'s.  Note that the colors represent the text authors.

\noindent {\bf victorian}:
Similarly, for the victorian dataset, available at \url{archive.ics.uci.edu/ml/datasets/Victorian+Era+Authorship+Attribution}. We use texts from up to 16 English-language authors from the Victorian era. Each text consists of $1{,}000$-word sequences obtained from a book written by one of these authors (we use the training dataset). The data was extracted and processed in~\cite{gungor2018fifty}. From each document, we extract a 10-dimensional vector using Gensim's Doc2Vec with the standard parameter settings again, and we assign the author id as color, as in prior work~\cite{kdd19}. We use 100 texts from each author, create one vertex for each text, and set the top $\theta \in \{0.25, 0.50, 0.75\}$ fraction of pairwise dot product edges as positive, and the remaining edges as negative. All graphs are unweighted and complete.

\subsection{Other experimental results}

\begin{table*}
\centering
\begin{tabular}{r|ccc}
\hline
Algorithm & \error & \imbalancecolor & \imbalancecolor \\
& & for 1/2 & for equality \\
\hline 
\localsearch &0.005&0.375&0.541\\
\pivot &0.009 &0.365  & 0.529 \\ 
\match + \localsearch &0.006	& 0	&0.518\\ 
\repmatch + \localsearch &0.070& 0	& 0 \\
\single & 0.828 & 0 & 0 \\
\randbaseline	&0.173 & 0 & 0 \\
\hline
\end{tabular}
\caption{Experimental results for \amazon, $C=4$ colors.  \label{tab:results-amazon-4}}
\end{table*}

In Table~\ref{tab:results-amazon-4} we report an overview of the results of the various algorithms for a dataset extracted from Amazon involving 250 vertex for each of 4 colors corresponding to the book categories {\it Literature \& Fiction, Nonfiction, Business \& Investing, Computers \& Internet}.

Similarly in Table~\ref{tab:results-reuters-8} we report the results for \reuters, $\theta=0.50$, $c=8$ colors.

\begin{table*}
\centering
\begin{tabular}{r|ccc}
\hline
Algorithm & \error & \imbalancecolor & \imbalancecolor \\
& & for 1/2 & for equality \\
\hline 
\localsearch &0.239&0.036&0.344\\
\pivot &0.35 &0.024& 0.298 \\ 
\match + \localsearch &0.251	& 0	&0.310\\ 
\repmatch + \localsearch &0.416& 0	& 0 \\
\single & 0.501 & 0 & 0 \\
\randbaseline	&0.500 & 0 & 0 \\
\hline
\end{tabular}
\caption{Experimental results for \reuters, $\theta=0.50$, $C=8$ colors. \label{tab:results-reuters-8}}
\end{table*}

Finally, in Table~\ref{tab:results-others} we report an evaluation of our algorithms in a variety of datasets and for different number of colors.

Notice how in all cases the results matches qualitatively the results reported in the main  paper.

\begin{table*}
\centering
\begin{tabular}{ll|rr}
\toprule
               &    &  \error  &  \error \\
                &    &  \localsearch   &  \repmatch + \\
                 &    &    & \localsearch \\
dataset & $C$ &        \\
\midrule
\reuters, $\theta=0.25$ & 2  &  0.096 &  0.230 \\
               & 4  &  0.120 &  0.244 \\
               & 8  &  0.133  &  0.252 \\
               & 16 &  0.146 &  0.255 \\
\hline               
\reuters, $\theta=0.50$ & 2  &  0.181 &  0.350 \\
               & 4  &  0.191 &  0.336 \\
               & 8  &  0.239&  0.416 \\
               & 16 &  0.258&  0.391 \\
\hline               
\reuters, $\theta=0.75$ & 2  &  0.188 & 0.199 \\
               & 4  &  0.211 &  0.227 \\
               & 8  &  0.237 &  0.250 \\
               & 16 &  0.220 &  0.250 \\
\hline               
\victorian, $\theta=0.25$ & 2  &  0.109 &  0.212 \\
               & 4  &  0.141 &  0.210 \\
               & 8  &  0.161 &  0.212 \\
               & 16 &  0.150&  0.222 \\
\hline               
\victorian, $\theta=0.50$ & 2  &  0.183 &  0.348 \\
               & 4  &  0.228 &  0.311 \\
               & 8  &  0.249 &  0.319 \\
               & 16 &  0.232  &  0.343 \\
\hline               
\victorian, $\theta=0.75$ & 2  &  0.203  &  0.237 \\
               & 4  &  0.225 &  0.245 \\
               & 8  &  0.218 &  0.246 \\
               & 16 &  0.215  &  0.250 \\
\bottomrule
\end{tabular}
\caption{Experimental results for various datasets and number of colors. \label{tab:results-others}}
\end{table*}

\paragraph{Additional baselines.}
We further experimented with two other greedy baselines.  First, we tried the following (unfair) greedy baseline: in an arbitrary order, iterate over the vertices, and for each vertex, add it to either the current cluster with most positive neighbors (if it exists) or to a singleton cluster.  More precisely, we assign the vertex to the best current cluster, if it is connected with more positive edges than negative edges to it, otherwise we leave the vertex as a singleton. Unsurprisingly, this unfair baseline is worse than all other unfair baselines we considered in terms of error and it has also a large imbalance, so we omit the results.

We also tested a fair greedy baseline for $\alpha=\frac{1}{2}$, for $C=2$: sort all pairs of different color vertices by distance in the Hamming space in an increasing order, and assign vertices to clusters of size $2$ with a greedy matching algorithm over this order. This creates fair clusters but again, we observe that this baseline to be close to that of \randbaseline and as such we omit the results.

\paragraph{Running time.}
All experiments have been conducted on commodity hardware. Each run of an algorithm completed in less than an hour. In our experiments, our fair algorithms have a running time in the same order of magnitude of that of the local search heuristic. For instance, for \reuters, $\theta=0.50$,  the ratio of mean running time of  \match+\localsearch  and  \repmatch + \localsearch w.r.t. \localsearch  was $90\%$ and $29\%$, respectively.  For \victorian, $\theta=0.50$ it was $123\%$ and $41\%$, respectively. 
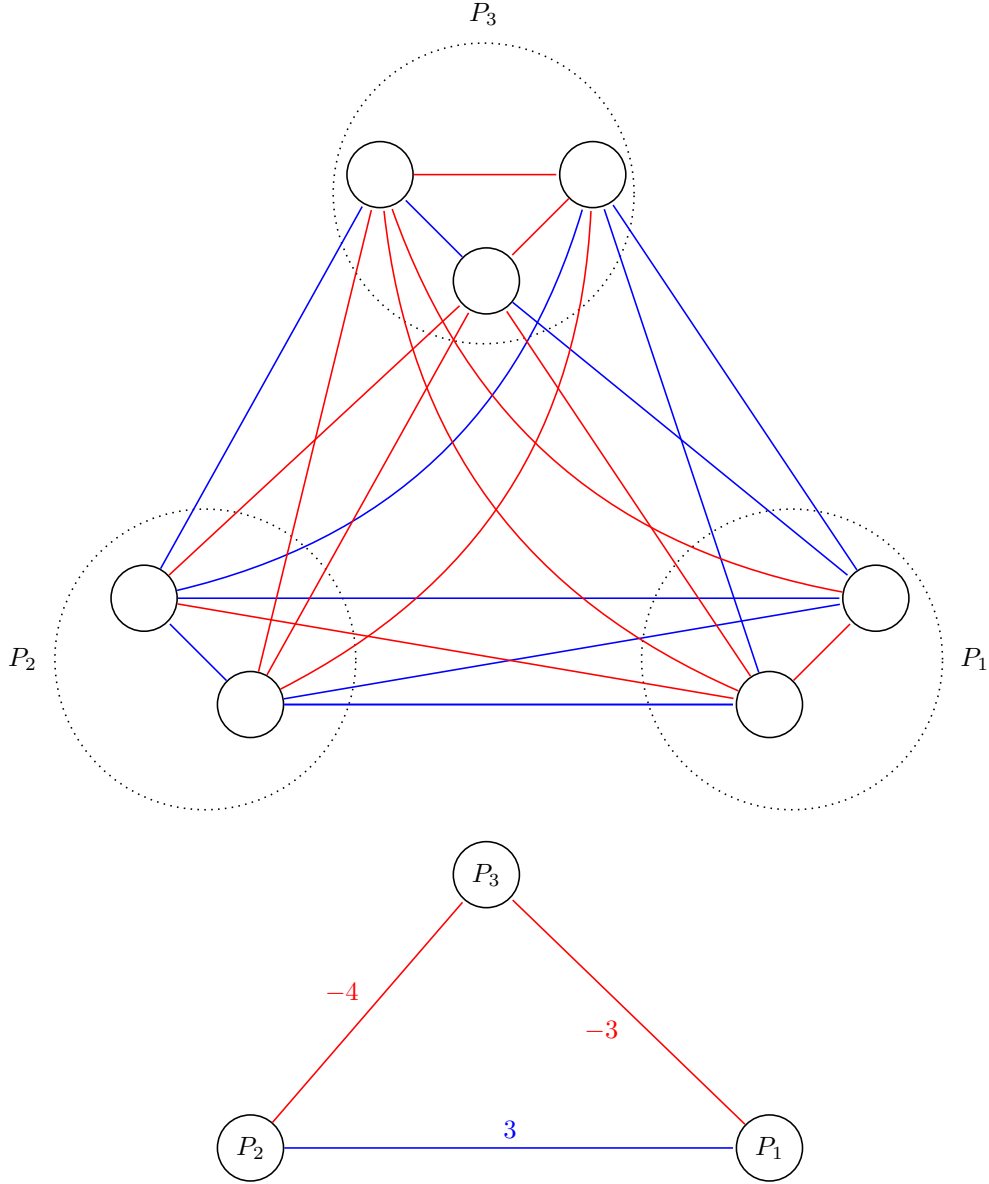
\begin{figure*}[htb]
\centering
\begin{tikzpicture}[-,>=stealth',shorten >=1pt,auto,node distance=2cm, , scale=0.20, semithick]

\node[state,color=black] (a) {};
\node[state,color=black] (b) [above left of = a] {};
\node[state,color=black] (c) [ right=6cm of a] {};
\node[state,color=black] (d) [above right of = c] {};

\node[state,color=black] (e) [above right=5cm and 2.5cm of a] {};
\node[state,color=black] (f) [above left of = e] {};
\node[state,color=black] (g) [above right of = e] {};

\node[state, dotted, minimum size=4cm] at (-3, 3) (p2) {};
\node[state, dotted, minimum size=4cm] at (36, 3) (p1) {};
\node[state, dotted, minimum size=4cm] at (15.5, 34) (p3) {};
\node[left=0.1cm of p2] {$P_2$};
\node[right=0.1cm of p1] {$P_1$};
\node[above=0.1cm of p3] {$P_3$};

\path[]
  (a) edge[blue] node {} (b)
  (b) edge[blue] node {} (d)
  (a) edge[blue] node {} (c)
  (a) edge[blue] node {} (d)
  (b) edge[blue] node {} (f)
  (d) edge[blue] node {} (g)
  (c) edge[blue] node {} (g)
  (a) edge[blue] node {} (c)
  (e) edge[blue] node {} (d)
  (e) edge[blue] node {} (f)
  (b) edge[blue, bend right, below] node {} (g)

  (c) edge[red] node {} (d) 
  (b) edge[red] node {} (c) 
  (a) edge[red] node {} (e) 
  (a) edge[red] node {} (f) 
  (a) edge[red, bend right, below] node {} (g) 
  (f) edge[red] node {} (g) 
  (b) edge[red] node {} (e) 
  (c) edge[red] node {} (e) 
  (g) edge[red] node {} (e)
  (c) edge[red, bend left, below] node {} (f)
  (d) edge[red, bend left, below] node {} (f)
  ;

\node[state, color=black] (pp2) [below=5cm of a] {$P_2$};
\node[state, color=black] (pp1) [below=5cm of c] {$P_1$};
\node[state, color=black] (pp3) [below=7cm of e] {$P_3$};
\path[]
  (pp2) edge[red] node {$-4$} (pp3)
  (pp2) edge[blue] node {$3$} (pp1)
  (pp1) edge[red] node {$-3$} (pp3)
;
\end{tikzpicture}
\caption{For partition $\CP = (P_1, P_2, P_3)$, graphs $G$ and $G^{\CP}$ are demonstrated; negative edges are red and positive edges are blue.  In this example,
$\inn(P_1) = \inn(P_3) = 1, \inn(P_2) = 0$ and $\out(P_i, P_j)$ are shown as a weight of edge $(P_i, P_j)$ in $G^\CP$. 
\label{fig:shift}}
\end{figure*}

\fi 

\end{document}